\documentclass{article}
\usepackage{amsmath}
\usepackage{tikz}
\usepackage{graphicx}
\usepackage{float}
\usepackage{amsthm}
\usepackage{amssymb}
\usepackage{booktabs}
\usepackage{bibentry}
\usepackage{array}
\usepackage{mathrsfs}
\usepackage{bm}
\usepackage{hyperref}
\usepackage{enumerate,algorithmicx,algorithm}
\usepackage{authblk}
\usepackage{algpseudocode}
\usepackage{geometry}

\allowdisplaybreaks

\DeclareMathOperator{\erf}{erf}
\DeclareMathOperator{\erfc}{erfc}

\newtheorem{mythe}{Theorem}[section]
\newtheorem{remark}{Remark}[section]

\title{Symmetry-preserving random batch Ewald method for constant-potential simulation of electrochemical systems}
\author[1,2]{Weihang Gao}
\author[2]{Qi Zhou\thanks{zhouqi1729@sjtu.edu.cn}}
\author[3]{Qianru Zhang}
\author[2]{Zhenli Xu\thanks{xuzl@sjtu.edu.cn}}

\affil[1]{Research Institute of Petroleum Exploration \& Development, PetroChina Corporation Limited, Beijing 100083, China}
\affil[2]{School of  Mathematical  Sciences, MOE-LSC and CMA-Shanghai, Shanghai  Jiao  Tong  University,  Shanghai 200240,  China}
\affil[3]{Hangzhou Huawei Enterprise Communication Technology Co., Ltd., Hangzhou 518129, China}

\date{}

\begin{document}

\maketitle

\begin{abstract}
Constant potential molecular dynamics simulation plays important role for applications of electrochemical systems, yet the calculation of charge fluctuation on electrodes remains a computational bottleneck. We propose a highly scalable, symmetry-preserving random batch Ewald (SRBE) algorithm to address this challenge. The SRBE algorithm deterministically computes the low-frequency components along the direction perpendicular to electrodes, while efficiently approximating the remaining components using random batch sampling. This approach simultaneously reduces charge and force fluctuations while satisfying the symmetry-preserving mean field condition in anisotropic systems with large aspect ratios. Numerical experiments on electrode/ionic liquid systems validate the high accuracy of the SRBE method in capturing dynamic charging processes and equilibrium electric double layer structures. The SRBE method achieves parallel efficiency improvements of up to two orders of magnitude compared with conventional FFT-based algorithms. These findings highlight its strong potential for enabling large-scale electrochemical simulations and its broad applicability to practical problems in the field.



\end{abstract}

\section{Introduction}
\label{sec::intro}
Electrochemical interfaces play a pivotal role in a wide range of physical and chemical processes, including electrolyte–electrode interactions \cite{sundararaman2022improving}, catalytic reactions \cite{stamenkovic2017energy}, and battery charge–discharge cycles \cite{mora2024high}.
Simulating charge distribution and potential variations at these interfaces is essential for understanding reaction efficiency, which largely depends on accurate modeling of charge transfer dynamics in molecular dynamics (MD) simulation. Typically, MD simulation can be grouped into constant charge (CQ) or constant potential (CP) simulations \cite{wang2014evaluation,yang2017reliability}. The CQ method assumes a fixed charge distribution, and offers high computational efficiency for static systems \cite{tee2023constant,nickel2024water}. However, this static assumption fails in dynamic environments involving charge transfer or fluctuating electrode potentials, where it cannot adequately capture the responsive electronic behavior in the systems \cite{bedrov2019molecular}. 
In contrast, CP simulations enforce a constant electrode potential, providing a more realistic depiction of charge redistribution and interfacial dynamics during electrochemical processes \cite{jeanmairet2022microscopic}.
Their ability to resolve electric double layer structures and capture nuanced charging behavior—especially under challenging conditions such as electrode-induced nanoconfinement, high electrolyte concentrations, and organic solvents—has driven major advances in the molecular-level understanding of electrochemical phenomena \cite{coles2019simulation,zeng2024constant}. Notably, CP methods are uniquely capable of modeling potential-dependent phenomena such as modeling of voltage-driven devices \cite{bonnet2012first}, reaction mechanism of CO$_2$ reduction on the electrochemical interface \cite{cheng2017full}, and potential-regulated atom cluster formation \cite{zhou2025constant}, which are inaccessible within the fixed-charge framework.
These simulations are critical for probing systems where electronic polarization and explicit potential control govern the behavior \cite{yu2023constant}, such as during faradaic reactions in batteries, potential-dependent electrocatalyst selectivity, and voltage-gated ion transport in nanoporous electrodes. Consequently, CP methods have become indispensable for investigating interfacial charge dynamics in diverse applications, including batteries, supercapacitors, and electrocatalytic systems \cite{merlet2013simulating,noh2019understanding,chen2020adding}.

The primary computational bottleneck in CP simulations is the substantial cost of continuously adjusting electrode charges at each timestep to maintain a fixed potential. Over the years, several approaches have been developed to address this challenge, including the Green function method \cite{raghunathan2007self,zeng2023molecular}, the induced charge computation (ICC*) method \cite{tyagi2010iterative}, and the charge fluctuation method \cite{siepmann1995influence,reed2007electrochemical}. Particularly, the charge fluctuation method has been successfully implemented in major software packages such as Metalwalls \cite{marin2020metalwalls}, LAMMPS \cite{tee2022fully,ahrens2022electrode}, and GROMACS \cite{bi2020molecular}. While acceleration techniques like the particle-particle particle-mesh (PPPM) method \cite{Hockney1988Computer} offer high accuracy and efficiency in electrostatic computing, they suffer from the communication bottlenecks that hinder large-scale simulations \cite{arnold2013comparison,ayala2021scalability}. 

To overcome this issue, the random batch Ewald (RBE) algorithm was recently proposed as a promising alternative \cite{jin2021random}. In the Fourier far-field part, the RBE method uses importance sampling approximations to replace the communication-intensive FFT procedure, thereby achieving efficient $O(N)$ complexity simulations. It also demonstrates excellent parallel scalability and has been successfully validated in various homogeneous systems \cite{liang2022superscalability,irbe2022jpca,gao2024rbmd}. Recent advancements have extended this approach to various physical applications, including random batch Monte Carlo method \cite{li2020random}, random batch sum-of-Gaussians method \cite{liang2023random,chen2025random}, and random batch list method \cite{zhang2025random,liang2021randomRBL,xu2024variance}.
However, its effectiveness breaks down for typical electrode/ionic liquid systems under the CP simulation. The inherent anisotropy and vacuum slab geometry common to these setups cause large variance in the stochastic updates such that the RBE method fails to satisfy the symmetry-preserving mean field (SPMF) condition \cite{hu2022symmetry} in the direction perpendicular to the electrode, rendering its direct application inefficient. The symmetry-preserving random batch Ewald (SRBE) method was first introduced for fully periodic systems \cite{gao2023screening}, but its extension to interface systems is not trivial.

In this paper, we develop a novel SRBE algorithm for accurate and efficient CP simulations of anisotropic electrode/ionic liquid systems.  This ``symmetry-preserving'' technique effectively eliminates the large variance that the original random batch methods encounter in anisotropic environments, and ensures the adherence to the SPMF condition \cite{gao2023screening}. The remaining interaction components are then handled efficiently using the conventional random batch sampling. We provide a rigorous theoretical proof demonstrating that, this hybrid approach not only preserves the favorable $O(N)$ linear scaling of the RBE method, but also strategically eliminates the primary source of variance in charge updates. Numerical results show that the SRBE  accurately captures the dynamics of electrode charging and reproduces the equilibrium structure of the electric double layer. Most notably, the SRBE algorithm delivers exceptional performance gains. Compared to the conventional PPPM algorithm, the SRBE achieves a six-fold speedup in weak scalability tests with 343 CPU cores. This advantage becomes even more dramatic in strong scalability tests, where it outperforms the PPPM by two orders of magnitude with 512 CPU cores. These findings establish SRBE as a powerful and robust tool, enabling large-scale, high-fidelity simulations of complex electrochemical systems that were previously computationally prohibitive.

The structure of this paper is as follows. Section~\ref{sec::overview} reviews the basic principles of the charge fluctuation method and the current mainstream PPPM algorithm used for its acceleration. Section~\ref{sec::SRBE} describes our novel SRBE algorithm in detail, together with rigorous theoretical analysis of the variance reduction property and discussion on its complexity. Section~\ref{sec::result} presents the CP simulation results for electrode/ionic liquid systems. Concluding remarks are made in Section~\ref{sec::conclusion}.

\section{Charge fluctuation on constant-potential electrodes}
\label{sec::overview}

In the electrolyte-supercapacitor system, electrolyte ions and  atoms making up the positive and negative electrodes are confined within a rectangular simulation box $\Omega=[0,L_x]\times[0,L_y]\times [0,H]$, denoting $\bm{l}=(L_x,L_y,H)$ the dimension scale. The system is composed of $N_{\text{ion}}$ electrolyte ions inside the box and $N_{\text{ele}}$ electrode atoms, located at $\{\bm{r}_j=(x_j,y_j,z_j)\}$ and $\{\bm{R}_i=(X_i,Y_i,Z_i)\}$, respectively. Assume that each atom is fixed on the electrode carrying a Gaussian charge distribution \cite{reed2007electrochemical}
\begin{equation}
\mathcal{Q}_{i}(\boldsymbol{r})=Q_{i} \left(\frac{1}{\pi\eta^2}\right)^{3/2} \exp \left(-\left|\boldsymbol{r}-\boldsymbol{R}_{i}\right|^{2}/ \eta^{2}\right).
\end{equation}
Here, $Q_{i}$ is the integrated charge and $\eta$ is the bandwidth.
In this system, there is a fixed electrode on both the upper and lower surfaces in the 
$z$-direction, with ions moving within the simulation box. Therefore, this problem is typically modeled as a quasi-2D periodic boundary condition simulation, where periodic boundary conditions are applied in the 
$x$ and $y$ directions, while the $z$-direction remains free. Meanwhile, the system often exhibits anisotropy, with the dimension in the free direction typically being slightly larger than those in the periodic directions. This property can be quantified as $H = \mu \max\{L_x, L_y\}$, where $\mu$ is a system-specific scaling factor greater than 1. In this setup, the total electrostatic potential energy of the system (in dimensionless units) reads 
\begin{equation}
U_{c}=\frac{1}{2} \int_{\Omega}\rho\left(\boldsymbol{r}^{\prime}\right)\mathrm{d} \boldsymbol{r}^{\prime} \int_{\mathbb{R}^2\times[0,H]}\frac{\widetilde{\rho}\left(\boldsymbol{r}^{\prime \prime}\right)}{\left|\boldsymbol{r}^{\prime}-\boldsymbol{r}^{\prime \prime}\right|}\mathrm{d} \boldsymbol{r}^{\prime \prime},
\label{tot_col}
\end{equation}
where 
\begin{equation}
    \label{eq::rho}
    \rho(\boldsymbol{r})=\sum_{i=1}^{N_{\text{ion}}}q_i\delta\left(\boldsymbol{r}-\boldsymbol{r}_{i}\right)+\sum_{j=1}^{N_{\text{ele}}}\mathcal{Q}_j(\bm{r}) 
\end{equation}
denotes the charge distribution in the simulation box and its periodization
\begin{equation}
    \label{eq::rho_peri}    \widetilde{\rho}(\boldsymbol{r})=\sum_{\bm{n}\in \mathbb{Z}^2\times\{0\}}\left[\sum_{i=1}^{N_{\text{ion}}}q_j\delta\left(\boldsymbol{r}-\boldsymbol{r}_{i}-\bm{n}\circ\bm{l}\right)+\sum_{j=1}^{N_{\text{ele}}}\mathcal{Q}_j(\bm{r}-\bm{n}\circ\bm{l})\right] 
\end{equation}
with `$\circ$' denoting the Hadamard product. Note that all the self-interaction at $\bm{r}^{\prime}=\bm{r}^{\prime\prime}$ are discarded in Eq.~\eqref{tot_col}.

During the CP simulation process, charges on the electrode atoms need to be updated at each simulation step to maintain a constant potential difference between the positive and negative electrodes \cite{tee2022fully}. Let $\boldsymbol{Q}=(Q_1,\cdots,Q_{N_{\text{ele}}})$, then the potential on the electrode atoms $\boldsymbol{\Psi}=(\Psi_1,\cdots,\Psi_{N_{\text{ele}}})$ can be represented by $\boldsymbol{\Psi} = \partial U_c/\partial {\boldsymbol{Q}}^T$. The potential energy $U_c$ can be written as the quadratic form of the electrode charge $\boldsymbol{Q}$, expressed by \cite{sitlapersad2024simple}
\begin{equation}
U_{c}(\boldsymbol{Q})=\frac{1}{2} \boldsymbol{Q}^{T} \boldsymbol{A} \boldsymbol{Q}-\boldsymbol{b}^{T} \boldsymbol{Q}+c,
\label{qua_Uc}
\end{equation}  
where $\boldsymbol{A}$ and $\boldsymbol{b}$ are determined by the configuration of electrolyte ions, and will be explicitly provided through the discussion in Section~\ref{sec::overview}. The constant term $c$ is independent of the electrode charges $\bm{Q}$ and does not contribute to the charge update process in CP simulations. Computing the electrode potential using Eq.~\eqref{qua_Uc} results in 
\begin{equation}
\boldsymbol{\Psi} =\boldsymbol{A Q}-\boldsymbol{b}. 
\end{equation}
One can split $\boldsymbol{\Psi}$ into the left and right electrode part, where the potential of the first half of the atoms maintains $\bar{\psi}-\Delta \psi/2 $ and the remaining atoms have the potential of $\bar{\psi}+\Delta \psi/2 $. Therefore, one rewrites the potential of electrode atoms in the form of constant potential difference $\Delta \psi$ so that
\begin{equation}
\begin{aligned}
&\boldsymbol{\Psi}=\bar{\psi} \boldsymbol{e}+\Delta \psi \boldsymbol{d}\\
&\boldsymbol{e}=(1,1,\cdots)\\
&\begin{matrix}\boldsymbol{d}=
 (\underbrace{-\frac{1}{2},\cdots,-\frac{1}{2},} \\{\ \qquad \text{Negative}}  \end{matrix}\begin{matrix}
 \underbrace{\frac{1}{2},\cdots,\frac{1}{2}}),\\ {\text{Positive}}  \end{matrix}
\end{aligned}
\end{equation}
which determines the distribution adjustment of $\boldsymbol{Q}$ at each step with
\begin{equation}
\boldsymbol{Q}=\boldsymbol{A}^{-1}(\bar{\psi} \boldsymbol{e}+\Delta \psi \boldsymbol{d}+\boldsymbol{b}).
\end{equation}
It can also be assumed that the charge distribution on the electrode satisfies the condition of electric neutrality $\boldsymbol{e}^T \boldsymbol{Q}=0$ \cite{scalfi2020charge}, which immediately derives

\begin{align}
\bar{\psi}&=-\frac{\boldsymbol{e}^T \boldsymbol{A}^{-1}(\Delta \psi \boldsymbol{d}+\boldsymbol{b})}{\boldsymbol{e}^T \boldsymbol{A}^{-1} \boldsymbol{e}}\label{psi_eq}\\ 
\boldsymbol{Q} & =\boldsymbol{J A}^{-1}(\Delta \psi \boldsymbol{d}+\boldsymbol{b}),
\label{Q_eq}
\end{align}
where 
\begin{equation}
\boldsymbol{J} \equiv \boldsymbol{I}-\left(\boldsymbol{A}^{-1} \boldsymbol{ee}^T\right) /\left(\boldsymbol{e}^T \boldsymbol{A}^{-1} \boldsymbol{e}\right) .    
\end{equation}
It is remarkable that satisfying the electric neutrality condition of the electrode plays a crucial role in simulations. The absence of this condition may lead to counterfactual phenomena. For example,  Ahrens-Iwers and Mei{\ss}ner \cite{ahrens2021constant} reported that in the empty gold capacitor, the surface charge on the electrode at the equilibrium  changes with introduced vacuum layer changing under the charge non-neutrality condition, which is inconsistent with the fact. For an intuitive illustration of the CP simulation model, a schematic is provided in Figure~\ref{cpm_model}.

\begin{figure}[h!]
		\centering
\includegraphics[width=1.0\linewidth]{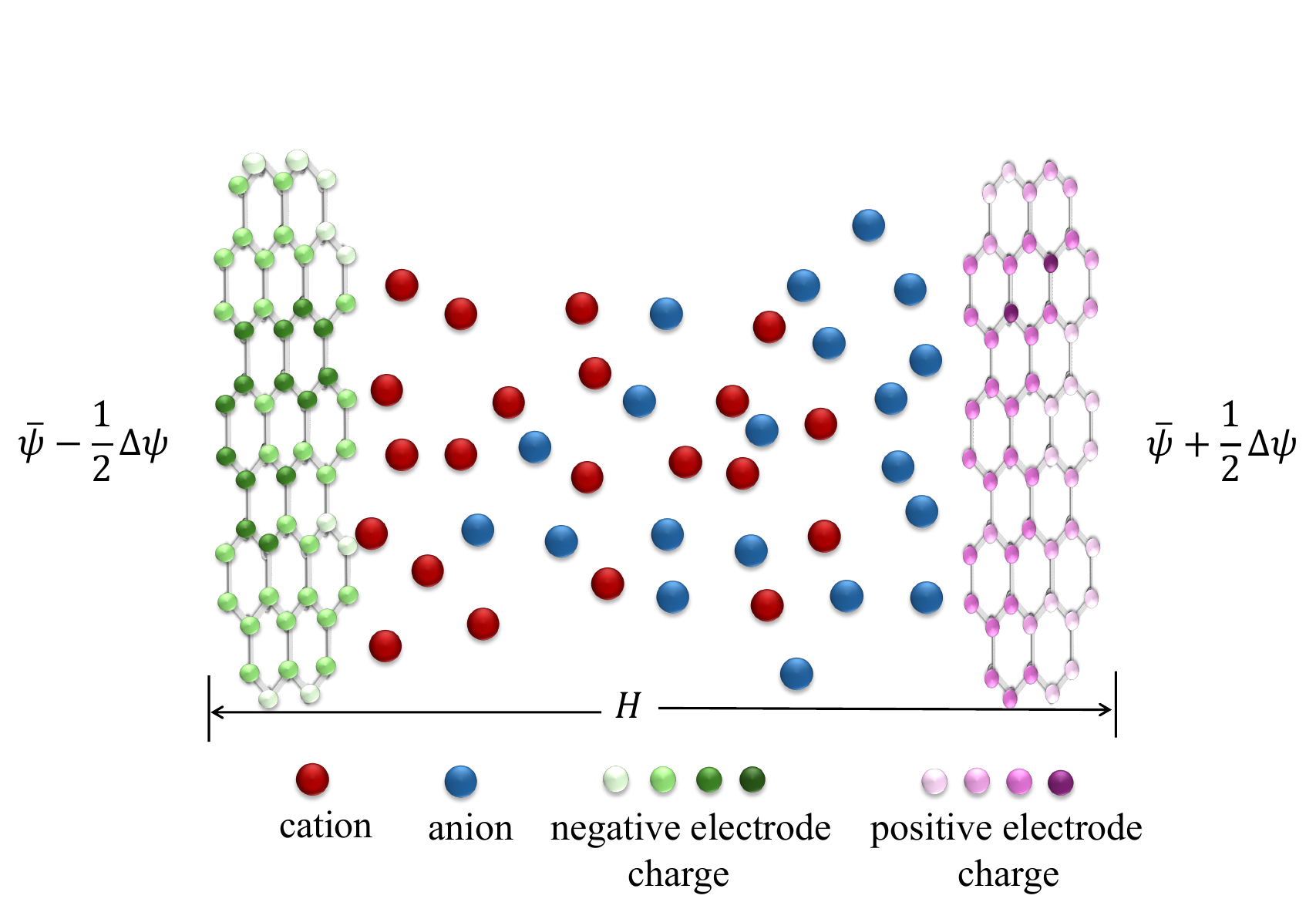}
		
        \caption{A schematic illustration of the electrode/ionic liquid system and its approximation system for the CP simulation. $\bar{\psi}-\frac{1}{2}\Delta\psi$ and $\bar{\psi}+\frac{1}{2}\Delta\psi$ are the potentials of the negative electrode and the positive electrode, respectively.  The different colors of the electrode charges indicate the dynamic changes of the electrode charges during the simulation process.}
		\label{cpm_model}
\end{figure}

In the CP simulation of the electrode/ionic liquid system, the calculation of the potential energy of Eq.~\eqref{tot_col} and the electrode charge adjustment of Eq.~\eqref{Q_eq} at each timestep requires the computation of pairwise interactions of long-range kernel. A classical acceleration technique is the use of the Ewald decomposition \cite{Ewald1921AnnPhys}, which splits the Coulomb interactions into the long-range part and short-range part as
\begin{equation}
    \label{eq::Ewald}
    \frac{1}{r}=\frac{\erf(\alpha r)}{r}+\frac{\erfc(\alpha r)}{r}.
\end{equation}
Here, $\erf(r)=\frac{2}{\sqrt{\pi}}\int_{0}^{r}\exp(-t^2)\mathrm{d}t$ denotes the error function and $\erfc(r)=1-\erf(r)$ be its complementary. The first part of Eq.~\eqref{eq::Ewald} is long-range and smooth, and is therefore treated in Fourier space, while the second part, due to its rapid decay, is directly truncated in the computation. Here $\alpha>0$ is a parameter that balances the convergence rates of both two components. Applying the Ewald decomposition to a quasi-2D system yields the well-known Ewald2D summation method \cite{parry1975electrostatic,zhonghanhu2014JCTC}, which decomposes Eq.~\eqref{tot_col} into $U_c=U_\mathcal{N}+U_\mathcal{F}^{\text{Q2D}}-U_{\text{corr}}$. Here the near part is calculated in the real space
\begin{equation}
    \label{eq::short_potential}
\begin{aligned}
U_\mathcal{N}&=\frac{1}{2} \sum_{i,j} \sum_{\boldsymbol{n}\in \mathbb{Z}^2\times \{0\}}{}^{\prime} q_{i} q_{j} \frac{\operatorname{erfc}\left(\alpha\left|\boldsymbol{r}_{ij}^{\bm{n}}\right|\right)}{\left|\boldsymbol{r}_{ij}^{\bm{n}}\right|} \\
&+\sum_{i,j} \sum_{\boldsymbol{n}\in \mathbb{Z}^2\times \{0\}} q_{i} Q_{j} \frac{\operatorname{erfc}\left(\alpha\left|\boldsymbol{r}_{i}-\boldsymbol{R}_{j}^{\bm{n}}\right|\right)-\operatorname{erfc}\left(\frac{1}{\eta}\left|\boldsymbol{r}_{i}-\boldsymbol{R}_{j}^{\bm{n}}\right|\right)}{\left|\boldsymbol{r}_{i}-\boldsymbol{R}_{j}^{\bm{n}}\right|}\\
&+\frac{1}{2} \sum_{i,j} \sum_{\boldsymbol{n}\in \mathbb{Z}^2\times \{0\}}{}^{\prime} Q_{i} Q_{j} \frac{\operatorname{erfc}\left(\alpha\left|\boldsymbol{R}_{ij}^{\bm{n}}\right|\right)-\operatorname{erfc}\left(\frac{1}{\sqrt{2}\eta}\left|\boldsymbol{R}_{ij}^{\bm{n}}\right|\right)}{\left|\boldsymbol{R}_{ij}^{\bm{n}}\right|},\\
\end{aligned}
\end{equation}
where $\bm{r}_{ij}^n=\bm{r}_i^n-\bm{r}_j^n$, $\boldsymbol{r}^{\bm{n}}:=\bm{r}+\bm{n}\circ\bm{l}$, and the prime represents the exclusion of the term of $i=j$ and $\bm{n}=\bm{0}$. The far part is evaluated by the Fourier transform,
\begin{equation}
    \label{eq::long_potential_q2D}
    U_\mathcal{F}^{\text{Q2D}}=\frac{1}{4\pi L_xL_y} \sum_{\bm{\dot{k}}\in 2\pi\mathbb{Z}^2\circ\bm{\dot{\ell}}^{-1}}\int_{k_z\in\mathbb{R}} \frac{4 \pi}{\dot{k}^{2}+k_z^2}\left|S_{1}([\boldsymbol{\dot{k}},k_z])+S_{2}([\boldsymbol{\dot{k}},k_z])\right|^{2} \exp \left(-\frac{\dot{k}^{2}+k_z^2}{4 \alpha^{2}}\right)\mathrm{d}k_z,
\end{equation}
where $\bm{\dot{\ell}}^{-1}=(L_x^{-1},L_y^{-1})$, $\dot{k}=|\bm{\dot{k}}|$ denotes its Euclidean $2$-norm, and the structure factors of ion and Gaussians are defined respectively by
\begin{equation}   
S_1(\boldsymbol{k})=\sum_{j=1}^{N_{\text{ion}}}q_j e^{i\boldsymbol{k}\cdot \boldsymbol{r}_j},\quad  and \ \ 
S_2(\boldsymbol{k})=\sum_{j=1}^{N_{\text{ele}}}Q_j e^{i\boldsymbol{k}\cdot \boldsymbol{R}_j}.
\end{equation}
To efficiently calculate the Fourier part Eq.~\eqref{eq::long_potential_q2D} by fast algorithms like the fast Fourier Transform (FFT), one employs the approximation 
\begin{equation}
    \label{eq::long_potential}
    U_\mathcal{F}=\frac{1}{2 L_xL_yL_z} \sum_{\bm{k}\in \mathcal{L}_k} \frac{4 \pi}{k^{2}}\left|S_{1}(\boldsymbol{k})+S_{2}(\boldsymbol{k})\right|^{2} \exp \left(-\frac{k^{2}}{4 \alpha^{2}}\right),
\end{equation}
where the Fourier integral in the non-periodic direction is discretized with the step size of $\Delta k_z=2\pi/L_z$, correspondingly the reciprocal Fourier mode is $ \mathcal{L}_k:=2\pi\mathbb{Z}^3\circ \bm{L}^{-1}\backslash\{\bm{0}\}$, $k=|\bm{k}|$. The corresponding correction term $U_{\text{self}}$ of this approach becomes \cite{yeh1999ewald}
\begin{equation}
    \label{eq::Self}
    U_{\text{corr}}=-\left(\frac{1}{\sqrt{2 \pi}\eta}-\frac{\alpha}{\sqrt{\pi}}\right) \sum_{i=1}^{N_{\text{ele}}} Q_{i}^{2}+\frac{\alpha}{\sqrt{\pi}} \sum_{i=1}^{N_{\text{ion}}} q_{i}^{2}-\frac{2\pi}{L_xL_yL_z} \left(\sum_{i} q_{i} z_i+\sum_{j} Q_{j} Z_j\right)^2,
\end{equation}
which is composed of the $i=j$ self-interaction and the $\bm{k}=\bm{0}$ Fourier mode.

An alternative physical interpretation of Eq.~\eqref{eq::long_potential} is the slab correction method \cite{yeh1999ewald,gingrich2010ewald,gingrich2010simulating}. This method artificially imposes periodicity in the non-periodic direction by introducing a vacuum layer. The core idea is that by making the vacuum layer sufficiently large, the electrostatic interactions between adjacent periodic images become negligible to a desired precision, thus enabling an accurate simulation of a non-periodic system within a periodic framework. Specifically, this corresponds to increasing the simulation box size in the $z$-direction from $H$ to $L_z=\lambda H=\lambda \mu\max\{L_x,L_y\}$, where $\lambda\ge 1$ is the zero-padding ratio of the vacuum layer. We then denote $V=L_xL_yL_z$ as the volume of the extended full-periodic simulation box. For a fixed truncation of the Fourier integral, reducing the discretization step size $\Delta k_z$ (i.e., increasing $\lambda$ and $L_z$) to improve accuracy inevitably increases the number of Fourier grid points that must be processed by the FFT, thus raising the overall computational cost. A detailed analysis of the optimal zero-padding ratio for quasi-2D systems can be found in \cite{gan2024fast, gan2025random, gao2024fast, liu2024optimal, gao2025accurate}, although these previous discussions on anisotropy have primarily focused on strongly confined systems (i.e., $H\ll L_x,L_y$). In the following, we discuss a 3D periodic simulation box after zero-padding, that is, we calculate Eq.~\eqref{eq::long_potential} instead of Eq.~\eqref{eq::long_potential_q2D} in the total electrostatic energy $U_c$.

CP simulations require the charge distribution $\bm{Q}$ on the electrode to be updated continuously according to Eq.~\eqref{Q_eq}. Using the Ewald decomposition, the quadratic matrix $\bm{A} = (A_{ij})$ and the vector $\bm{b} = (b_1,\cdots,b_{N_\text{ele}})^T$ can similarly be written as
\begin{equation}
    \label{eq::A}
    \begin{aligned}
        A_{i j}=&\frac{1}{V} \sum_{\bm{k}\in \mathcal{L}_k} \frac{4 \pi}{k^{2}} e^{i \boldsymbol{k} \cdot \boldsymbol{R}_{i j}}e^{-k^{2} / (4 \alpha^{2})}+\sum_{\boldsymbol{n}\in \mathbb{Z}^2\times \{0\}}{}^{\prime}\frac{\operatorname{erfc}\left(\alpha\left|\boldsymbol{R}_{i j}^{\bm{n}}\right|\right)-\operatorname{erfc}\left(\frac{1}{\sqrt{2}\eta}\left|\boldsymbol{R}_{i j}^{\bm{n}}\right|\right)}{\left|\boldsymbol{R}_{i j}^{\bm{n}}\right|}\\
&+2 \delta_{i j}\left(\frac{1}{\sqrt{2 \pi}\eta}-\frac{\alpha}{\sqrt{\pi}}\right)+\frac{2\pi}{V}Z_i Z_j, \quad\quad i,j=1,2,\cdots,N_{\text{ele}},\\
    \end{aligned}
\end{equation}
and
\begin{equation}
\begin{aligned}
        \label{eq::b}
    b_{j} :=& b_{\mathcal{F},j}-b_{\text{corr},j}+b_{\mathcal{N},j}+\Delta \psi\\
    =&-\frac{1}{2 V} \sum_{\bm{k}\in \mathcal{L}_k} \frac{4 \pi}{k^{2}}\left[e^{i \boldsymbol{k} \cdot \boldsymbol{R}_{j}} S_{1}(\boldsymbol{k})^{*}+e^{-i \boldsymbol{k} \cdot \boldsymbol{R}_{j}} S_{1}(\boldsymbol{k})\right] e^{-k^{2} / (4 \alpha^{2})}+\frac{4\pi Z_j}{V}\sum_{i=1}^{N_{\text{ion}}}q_iz_{i}\\
& -\sum_{i=1}^{N_\text{ion}} \sum_{\boldsymbol{n}\in \mathbb{Z}^2\times \{0\}}q_{i} \frac{\operatorname{erfc}\left(\alpha\left|\boldsymbol{R}_{j}-\boldsymbol{r}_{i}^{\bm{n}}\right|\right)-\operatorname{erfc}\left(\frac{1}{\eta}\left|\boldsymbol{R}_{j}-\boldsymbol{r}_{i}^{\bm{n}}\right|\right)}{\left|\boldsymbol{R}_{j}-\boldsymbol{r}_{i}^{\bm{n}}\right|} +\Delta{\psi}.\\
\end{aligned}
\end{equation}
where $S_1(\boldsymbol{k})^*$ being the conjugate of $S_1(\boldsymbol{k})$.

Since the electrode atoms are fixed, matrix $\bm{A}$ is constant and needs to be computed only once at initialization. The computational bottleneck is therefore the calculation of vector $\bm{b}$, whose direct evaluation involves an $O(N^2)$ cost due to ion-Gaussian interactions, where $N=N_{\text{ion}}+N_{\text{ele}}$ denotes the total number of particles. To overcome this, the PPPM algorithm is typically employed, reducing the complexity of computing $\bm{b}$ to $O(N\log N)$ \cite{tee2022fully,ahrens2021constant}. In the following, we develop the symmetric-preserving random batch Ewald method to the computation of electrolyte ion dynamics and electrode charge fluctuations by avoiding the use of communication-intensive FFT and achieving $O(N)$ complexity.

\section{Symmetry-preserving random batch Ewald method}
\label{sec::SRBE}
In this section, we propose a novel, efficient algorithm for CP simulations of electrode/ionic liquid systems, achieved by extending the symmetry-preserving random batch Ewald (SRBE) method \cite{gao2023screening}. Our approach deterministically computes the important low-frequency Fourier modes normal to the electrodes, while the remaining contributions are handled stochastically via a random mini-batch importance sampling strategy. This hybrid strategy efficiently accelerates the update of both electrode charges and electrolyte ion positions. 

During the evolution of the electrode/ionic liquid system, the update of the charge distribution on the electrode in Eq.~\eqref{Q_eq}  is computationally expensive and requires efficient acceleration. The main computational cost lies in the evaluation of vector $\bm{b}$ at each time step, whose elements take the form given in Eq.~\eqref{eq::b}. Indeed, the far part of $b_j$ reads
\begin{equation}
b_{\mathcal{F},j} =-\frac{1}{2 V} \sum_{\bm{k}\in \mathcal{L}_k} \frac{4 \pi}{k^{2}}\left[e^{i \boldsymbol{k} \cdot \boldsymbol{R}_{j}} S_{1}(\boldsymbol{k})^{*}+e^{-i \boldsymbol{k} \cdot \boldsymbol{R}_{j}} S_{1}(\boldsymbol{k})\right] e^{-k^{2} / (4 \alpha^{2})},\quad j=1,2,\cdots,N_{\text{ele}}.
\label{eq::b_j_far}
\end{equation}
Additionally, ions carry fixed charges and can move freely within the electrolyte. They interact with each other and are influenced by the charge distribution on the electrode. After applying the Ewald decomposition as mentioned in Section~\ref{sec::overview}, the far component of their interactions is given by the negative gradient of Eq.~\eqref{eq::long_potential},
\begin{equation}
\label{eq::force}
\begin{aligned}
  \boldsymbol{F}_{\mathcal{F},i}&= -\nabla_{\boldsymbol{r}_i}U_{\mathcal{F}}(\{\boldsymbol{r}_i\},\boldsymbol{Q})\\
  &=-\sum_{\bm{k}\in \mathcal{L}_k}\frac{4\pi q_i\boldsymbol{k}}{Vk^2}e^{-k^2/(4\alpha^2)}\Big(\text{Im}(S_1(\boldsymbol{k})e^{-i\boldsymbol{k}\cdot\boldsymbol{r}_i})+\text{Im}\left(S_2(\boldsymbol{k})e^{-i\boldsymbol{k}\cdot\boldsymbol{r}_i}\right)\Big).
\end{aligned}
\end{equation}

The electrode charge fluctuation and electrolyte ion dynamics are usually treated by the PPPM method \cite{ahrens2022electrode,ahrens2021constant}.
However, in the case of $L_z>\max\{L_x,L_y\}$, the Fourier modes in reciprocal space exhibit a finer resolution along the $z$-direction. These modes, being closer to the origin, capture the dominant long-range electrostatic contributions.  A classical acceleration method involves employing the random batch Ewald (RBE) algorithm \cite{jin2021random,liang2022superscalability,liang2022random} to accelerate computations in the Fourier space. However, directly applying the RBE algorithm by importance sampling in the Fourier space encounters a critical limitation: it fails to adequately sample Fourier space frequencies near the origin along the $z$-direction. This results in a large variance in the charges on the electrodes, potentially destabilizing or even causing the failure of the simulated dynamical process, which will be shown in the simulation section. Fortunately, this issue can be addressed by modifying the RBE algorithm introducing the SPMF condition proposed in previous work \cite{gao2023screening}. Taking $ \boldsymbol{F}_{\mathcal{F},i}$ as the example, $ \boldsymbol{F}_{\mathcal{F},i}$ can be divided into two components: the ion-ion interaction $ \boldsymbol{F}_{\mathcal{F},i}^{\text{ion}}$ and the ion-electrode interaction $ \boldsymbol{F}_{\mathcal{F},i}^{\text{ele}}$. If $ \boldsymbol{F}_{\mathcal{F},i}^{\text{ion}}$ is expressed in terms of force on ion-pairs, it can be written as
\begin{equation}
\begin{aligned}   \boldsymbol{F}_{\mathcal{F},i}^{\text{ion}}&=\sum_{j=1}^{N_{\text{ion}}}q_i\boldsymbol{f}_{ij}^{\text{ion}},\\
\boldsymbol{f}_{ij}^{\text{ion}}&=-\sum_{\bm{k}\in \mathcal{L}_k}\frac{4\pi\boldsymbol{k}}{Vk^2}e^{-k^2/(4\alpha^2)}\text{Im}(e^{-i\boldsymbol{k}\cdot\boldsymbol{r}_{ij}}).
\end{aligned} 
\end{equation} 
Then the SPMF condition $\langle\cdot \rangle_{\text{sp}}$ for the force $\boldsymbol{f}_{ij}^{\text{ion}}$ on the $z$ direction by performing the integration of $\boldsymbol{f}_{ij}^{\text{ion}}$ in the $x$ and $y$ directions \cite{gao2023screening} is
\begin{equation}
\langle\boldsymbol{f}_{ij}^{\text{ion}}\rangle_{\text{sp}}=-\frac{4\pi}{V}\sum_{k_z\neq 0}\frac{e^{-k_z^2/(4\alpha^2)}}{k_z}\text{Im}(e^{-i{k}_z {z}_{ij}}).
\label{eq::spmf_force}
\end{equation}  
Note that in the Fourier space, the long-range interaction force between particles can accurately capture the electrostatic correlation effects in bulk regions only if it satisfies Eq.~\eqref{eq::spmf_force}, which represents the SPMF condition. Electrostatic correlation effects are highly sensitive to the distribution and fluctuations of the charges of free particles. Therefore, a similar approach can be employed to address the challenge issue of large fluctuations in electrode charges under CP simulations in this study. From the expression of Eq.~\eqref{eq::spmf_force}, it can be understood that accurately calculating all Fourier frequency contributions in the $z$-direction of $\boldsymbol{F}_{\mathcal{F},i}^{\text{ion}}$ for every ions ensures precise satisfaction of the SPMF condition. Furthermore, due to the rapid convergence of the exponential term, it suffices to accurately compute the Fourier frequencies in the $z$-direction that are closest to the origin. Specifically, the low-frequency Fourier modes along the $z$-direction are computed deterministically. The number of these modes is determined by the aspect ratio $\lambda\mu = L_z / \max\{L_x,L_y\}$ with $O(1)$ scaling. The remaining high-frequency components are then efficiently evaluated stochastically using the  RBE method. The modified RBE method, which is to satisfy the SPMF condition is called as the SRBE method. 

By comparing Eqs.~\eqref{eq::b_j_far} and \eqref{eq::force}, it can be observed that the expression for $b_{\mathcal{F},j}$ shares similarities with the expression for the long-range force in terms of the frequency $k$. More importantly, the value of $b_{\mathcal{F},j}$ at each simulation step is directly involved in the calculation of electrode charges, which is highly sensitive for ensuring the proper progression of CP simulations. This highlights the critical importance of accurately computing $b_{\mathcal{F},j}$. Therefore, to accelerate CP simulations while maintaining their accuracy, we adopt the concept from the SRBE algorithm, which precisely calculates the low-frequency contribution in the $z$-direction and stochastically approximates the high-frequency components, to update $b_{\mathcal{F},j}$ in the electrode charge calculations. As a result, the far component of  $\boldsymbol{b}_j$  can be rewritten as 
\begin{equation}
b_{\mathcal{F},j}^{*} = -\sum_{\bm{k}\in \mathcal{I}} \frac{4 \pi}{V }\frac{\text{Re}\left(e^{-i \boldsymbol{k} \cdot \boldsymbol{R}_{j}} S_{1}(\boldsymbol{k})\right)}{k^2} e^{-k^2/(4\alpha^2)} -\frac{G}{P} \sum_{\ell=1}^P \frac{4\pi}{V}\frac{\text{Re}\left(e^{-i \boldsymbol{k}_{\ell} \cdot \boldsymbol{R}_{j}} S_{1}(\boldsymbol{k}_{\ell})\right)}{k_{\ell}^{2}},
\label{eq::vec_b_srbe}
\end{equation}
and the SRBE force for electrolyte ion $i$ reads
\begin{equation}
    \label{eq::F_SRBE}
    \begin{aligned}
\boldsymbol{F}_{\mathcal{F},i}^{*}=&-\sum_{\bm{k}\in \mathcal{I}}\frac{4\pi q_i\boldsymbol{k}}{Vk^2}e^{-k^2/(4\alpha^2)}\Big[\text{Im}(S_1(\boldsymbol{k})e^{-i\boldsymbol{k}\cdot\boldsymbol{r}_i})+\text{Im}\left(S_2(\boldsymbol{k})e^{-i\boldsymbol{k}\cdot\boldsymbol{r}_i}\right)\Big]\\
&-\sum_{\ell=1}^{P}\frac{G}{P}\frac{4\pi q_i \bm{k}_\ell}{Vk_\ell^2}\Big[\text{Im}(S_1(\boldsymbol{k}_\ell)e^{-i\boldsymbol{k}_\ell\cdot\boldsymbol{r}_i})+\text{Im}\left(S_2(\boldsymbol{k}_\ell)e^{-i\boldsymbol{k}_\ell\cdot\boldsymbol{r}_i}\right)\Big],
\end{aligned}
\end{equation}
where $\mathcal{I}=\{(0,0,2\pi m_z/L_z),\ 1\le|m_z|\le M,\ m_z\in\mathbb{Z}\}$ with $M$ denoting a truncation parameter depending on the aspect ratio $\lambda\mu$. The Fourier modes in the latter summation are randomly sampled in $\mathcal{L}_k \backslash \mathcal{I}$ from the  probability distribution $\mathcal{P}(\bm{k})$,
\begin{equation}
    \label{eq::Prob}
    \{\boldsymbol{k}_{\ell}\}_{\ell=1}^P\sim \mathcal{P}(\bm{k}):=e^{-k^2/(4\alpha^2)}/G,\quad \boldsymbol{k}_{\ell}\in \mathcal{L}_k\backslash\mathcal{I},
\end{equation}
where the normalization constant is given by
\begin{equation}
    \label{eq::NCons}
G=\sum_{\bm{k}\in\mathcal{L}_k\backslash\mathcal{I}}e^{-k^2/(4\alpha^2)}.
\end{equation}

To sample from $\mathcal{P}(\bm{k})$, we employ the Metropolis algorithm. The Metropolis procedure samples from the three-dimensional Gaussian distribution as follows. First, due to the separability of the Gaussian distribution, proposal values $\bm{m}=(m_x,m_y,m_z)$ are independently sampled along each dimension as $ u_d^* \sim \mathcal{N}\left(0, \alpha^2 L_d^2 / (2\pi)^2 \right)$, $d\in \{x,y,z\}$. Then, the proposal is rounded to the nearest integer as $ m^*_d = \mathrm{round}(u_d^*)$. Subsequently, the acceptance probability $q_d(m_d^*|m_d)
$ is defined as
\begin{equation}
    \label{eq::prob_proposal}
    q_d(m_d^*|m_d)=q_d(m_d^*):=\int_{m_d^*-1/2}^{m_d^*+1/2}\sqrt{\frac{\pi}{\alpha^2L_d^2}}e^{-\pi^2x^2/(\alpha L_d)^2}dx, 
\end{equation}
and the proposal $\bm{k}^*=2\pi \bm{m}^*\circ \bm{L}^{-1}$, if belonging to the high-frequency components, is then accepted with the following acceptance probability
\begin{equation}
    \label{eq::accept}
    a(\bm{m}^*|\bm{m})=\min\left\{1,\frac{\mathcal{P}(\bm{k^*})q_x(m_x)q_y(m_y)q_z(m_z)}{\mathcal{P}(\bm{k})q_x(m_x^{*})q_y(m_y^{*})q_z(m_z^{*})}\right\}.
\end{equation}
Since both the proposal distribution and the transition probabilities are derived from continuous Gaussian distributions or their truncated counterparts over small intervals, the Metropolis procedure achieves a very high acceptance rate.

In summary, each step in the simulation updates the positions and velocities of all particles along with the charges on the electrodes. The complete flowchart for this process, accelerated by the proposed SRBE algorithm, is detailed in Algorithm \ref{algo_srbe}.

\begin{algorithm}[H]
\caption{Symmetry-preserving random batch Ewald algorithm}\label{algo_srbe}
\begin{algorithmic}[1] 
\Require Select cutoff $r_c$, time step $\Delta t$, total simulation steps $N_{\text{step}}$, batch size $P$, size of low-frequency Fourier mode $M$, preset constant potential difference $\Delta\psi$. 
\State Initialize positions, velocities, and charges of all electrolyte ions and electrode charges. 
\For{$n=1,2,\cdots, N_{\text{step}}$}
\State  Sample $\boldsymbol{k}\sim G^{-1}e^{-|\boldsymbol{k}|^2/(4\alpha^2)}$ with $\boldsymbol{k}\neq \boldsymbol{0}$ for  $P$ frequencies $\{\boldsymbol{k}_{\ell}\}\in \mathcal{L}_k\backslash\mathcal{I}$.
\State Compute the vector $\boldsymbol{b}$ by Eq.~\eqref{eq::vec_b_srbe}. Update the electrode charge $\boldsymbol{Q}$ by Eq.~\eqref{Q_eq}.
\State For each ion, compute the short-range force by direct truncation, and compute the long-range force by Eq.~\eqref{eq::F_SRBE}. 
\State Integrate Newton's equation with proper numerical scheme and thermostat.
\EndFor
\State \textbf{Return:} Configurations $\{(\{\bm{r}_i^n\},\{\bm{v}_i^n\},\bm{Q}^n)\}_{n=1}^{N_{\text{step}}}$ for each discrete time $t_n=n\Delta t$.
\end{algorithmic}
\end{algorithm}

To demonstrate the applicability and efficiency of the SRBE method, we provide several theoretical results, including an analysis of variance reduction, convergence analysis, and an overall complexity analysis of the algorithm.

First, we analyze $b_{\mathcal{F},j}^{*}$ for each electrode atom $j$ at each time step for the electrode charge fluctuation. Let 
\begin{equation}
{\xi}_j:={b}_{\mathcal{F},j}^{*}-{b}_{\mathcal{F},j}, 
\end{equation}
and under the Debye H\"{u}ckel (DH) theory (see \cite{hansen2013theory} and also Appendix \ref{app::DH}), Theorem \ref{thm::var_bj_srbe} provides the variance bound of $\xi_j$.

\begin{mythe}
\label{thm::var_bj_srbe}
For the SRBE method in CP simulation of the electrode/ionic liquid system, the fluctuation ${\xi}_j$ has zero expectation. Furthermore, under the DH theory and charge neutrality condition, the variance of $\xi_j$ of the SRBE method scales as $O([\exp({-\omega^2/(4\alpha^2))}+\erfc(\omega/(2\alpha))]/P)$, where $\omega$ is a parameter depending on the box scale and truncation parameter $M$.
\end{mythe}

\begin{proof}
    By the electric neutrality of the ions of the electrolytes, it holds that
\begin{equation}
    \frac{\text{Re}\left(e^{-i \boldsymbol{k} \cdot \boldsymbol{R}_{j}} S_{1}(\boldsymbol{k})\right)}{k^2}\leq C,
\end{equation}
where $C$ is a constant from the DH assumption. One has 
\begin{equation}
\begin{aligned}
\mathbb{E}\left|{\xi}_j\right|^2&=\frac{1}{P}\left(\sum_{\boldsymbol{k}\in \mathcal{L}_k\backslash\mathcal{I}}\frac{16\pi^2G}{V^2}e^{-k^2/(4\alpha^2)}C^2 -|{b}_{\mathcal{F},j}^{\text{HF}}|^2\right)\\
&\leq \frac{G}{P} \frac{\widetilde{C}_{I}16\pi^2}{V^2}\int_{\omega}^{\infty} \frac{V}{(2\pi)^3}{e^{-k^2/(4\alpha^2)}}4\pi k^2 d k\\
&\leq \frac{16\alpha^2\widetilde{C}_{I}G}{PV}\left[\omega e^{-\omega^2/(4\alpha^2)}+\alpha\sqrt{\pi} \cdot\text{erfc}\left(\frac{\omega}{2\alpha}\right)\right],
\end{aligned}
\label{eq::charge_Var}
\end{equation}
where
\begin{equation}
    \label{eq::charge_HF}
{b}_{\mathcal{F},j}^{\text{HF}}:=-\sum_{\boldsymbol{k}\in \mathcal{L}_k\backslash\mathcal{I}} \frac{4 \pi}{V }\frac{\text{Re}\left(e^{-i \boldsymbol{k} \cdot \boldsymbol{R}_{j}} S_{1}(\boldsymbol{k})\right)}{k^2} e^{-k^2/(4\alpha^2)},
\end{equation}
and \begin{equation}
    \label{eq::omega}
    \omega=\min\left\{\frac{2\pi}{\max\{L_x,L_y\}},\frac{2\pi (M+1)}{L_z}\right\}
\end{equation}
denotes the minimal Fourier mode among the whole sampling space. The second inequality in Eq.~\eqref{eq::charge_Var} arises from approximating the trapezoidal sum by a continuous integral, a technique widely used in the error analysis of Ewald summation \cite{kolafa1992cutoff}. The validity of this approximation has been adequately discussed in the Appendix B of \cite{LIANG2025101759}. When selecting $\alpha=\rho_{\text{ion}}^{1/3}$ be an $O(1)$ parameter, one finds that the cost in short-range interaction becomes linear and also $G=O(V)$ \cite{jin2021random}, which proves the result.
\end{proof}

\begin{remark}
    As can be seen from Eqs. \eqref{eq::charge_Var} and \eqref{eq::omega}, if the symmetry-preserving technique is not applied in the RBE sampling (corresponding to $M = 0 $), then where $L_z \gg L_x$ or $ L_y $, the $\erfc$ factor in the variance will be close to 1. In contrast, the application of the SRBE method significantly reduces the $\erfc$ factor, achieving variance reduction at a Gaussian decay rate. It is also important to note that once the truncation parameter $M$ exceeds the aspect ratio $\lambda\mu$, the Fourier modes in the periodic directions begin to dominate the decay of the $\erfc$ factor, and further increasing $M$ becomes unnecessary.
\end{remark}

Then, in a similar way, we can analyze the approximation of the random force in a single step for the ion dynamics in the electrolyte. Define the fluctuation for the long-range force on atom $i$ by
\begin{equation}
\boldsymbol{\chi}_i:=\boldsymbol{F}_{\mathcal{F},i}^{*}-\boldsymbol{F}_{\mathcal{F},i},    
\end{equation}
and Theorem \ref{thm::1/P} provides an analogous variance reduction result when the SRBE method is employed. Again, the introduction of the symmetric-preserving technique yields a reduction in the variance prefactor, with a Gaussian decay rate with respect to the truncation parameter $M$.

\begin{mythe}
\label{thm::1/P}
The force fluctuation of each ion $\boldsymbol{\chi}_i$ has zero expectation. Suppose that the ion density $ \rho_{\emph{ion}} = N_{\emph{ion}} / V $ and the electrode atom number density $ \rho_{\emph{ele}} = N_{\emph{ele}} / V $ are fixed. Under the DH assumption, the SRBE method yields a variance of the stochastic force that scales as $O(\erfc(\omega/(2\alpha))/P)$, with a prefactor independent of the system size $ N_{\emph{ion}} $ and 
$ N_{\emph{ele}}$. 
\end{mythe}

\begin{proof}
    Based on the derivation under the DH approximation, the structure factor term can be bounded, that is, for every Fourier modes $\bm{k}\neq \bm{0}$,
\begin{equation}
|\operatorname{Im}\left(e^{-i \boldsymbol{k} \cdot \boldsymbol{r}_{\boldsymbol{i}}} S_1(\boldsymbol{k})\right)| \leq C_1, \quad |\operatorname{Im}\left(e^{-i \boldsymbol{k} \cdot \boldsymbol{r}_{\boldsymbol{i}}} S_2(\boldsymbol{k})\right)| \leq C_2,
\end{equation}
where $C_1$ and $C_2$ are constants independent of $ N_{\text{ion}} $ and 
$ N_{\text{ele}}$. Then the variance of the random force provided by the SRBE can be estimated by
\begin{equation}
\label{eq::F_Var_estimate}
\begin{aligned}
\mathbb{E}\left|\boldsymbol{\chi_i}\right|^2&\leq  \frac{1}{P}\left(\sum_{\boldsymbol{k} \in \mathcal{L}_k\backslash\mathcal{I}} \frac{\left(4 \pi q_i\right)^2 G}{V^2 k^2} e^{-k^2 /(4 \alpha^2)}(C_1+C_2)^2-\left|\boldsymbol{F}_{\mathcal{F},i}^{\text{HF}}\right|^2\right)\\
&\leq \frac{G}{P} \frac{C_{I}(4\pi q_i)^2}{V^2}\int_{\omega}^{\infty} \frac{V}{(2\pi)^3}\frac{e^{-k^2/(4\alpha^2)}}{k^2}4\pi k^2 d k\\
&= \frac{8\sqrt{\pi}q_i^2C_{I}\alpha G}{PV}\erfc\left(\frac{\omega}{2\alpha}\right),
\end{aligned}
\end{equation}
where
\begin{equation}
    \label{eq::Force_HF}
\boldsymbol{F}_{\mathcal{F},i}^{\text{HF}}:=-\sum_{\boldsymbol{k}\in \mathcal{L}_k\backslash\mathcal{I}}\frac{4\pi q_i\boldsymbol{k}}{Vk^2}e^{-k^2/(4\alpha^2)}\Big(\text{Im}(S_1(\boldsymbol{k})e^{-i\boldsymbol{k}\cdot\boldsymbol{r}_i})+\text{Im}\left(S_2(\boldsymbol{k})e^{-i\boldsymbol{k}\cdot\boldsymbol{r}_i}\right)\Big)
\end{equation}
is the force contribution of the high-frequency Fourier modes in $\mathcal{L}_k\backslash \mathcal{I}$.
The constant $C_I/(C_1+C_2)^2$ is derived from the integral approximation. Similar with the proof in Theorem \ref{thm::var_bj_srbe}, we can obtain, $\mathbb{E}\left|\boldsymbol{\chi}_i\right|^2=O(\text{erfc}(\omega/(2\alpha))/P)$.
\end{proof}

We now discuss the effect of the stochastic algorithm on the overall dynamics of the CP simulation. In the canonical (NVT) ensemble, a common method for temperature control is the stochastic Langevin thermostat \cite{hoover1985canonical}, whose mathematical expression for the electrode/ionic liquid system is provided by
\begin{equation}
\left\{
\begin{aligned}
&d\boldsymbol{r}_i=\boldsymbol{v}_i d t\\
&m_id\boldsymbol{v}_i=(\boldsymbol{F}_i-\gamma \boldsymbol{v}_i)dt + \sqrt{2\gamma/\beta}d\boldsymbol{W}_i\\
&\boldsymbol{Q} =\boldsymbol{J A}^{-1}(\Delta \psi \boldsymbol{d}+\boldsymbol{b}),
\end{aligned}
\right.
\label{langevin_ori}
\end{equation}
where $m_i$ is the mass of the $i$-th particle, $\gamma$ is the reciprocal charactersitic time with respect to the thermostat, $\beta=1/k_BT$ with the temperature $T$ and Boltzmann constant $k_B$, and $\{\boldsymbol{W}_i\}$ are the independent and identically distributed Wiener processes. When the system is evolved using the SRBE method as formulated in Eq.~\eqref{langevin_ori}, the configuration $(\{\bm{r}_i^*
\},\{\bm{v}_i^*\},\bm{Q}^*)$ follows the dynamical process given by
\begin{equation}
\left\{
\begin{aligned}
&d\boldsymbol{r}^*_i=\boldsymbol{v}^*_i dt\\
&m_i d\boldsymbol{v}^*_i=(\boldsymbol{F}_i+\boldsymbol{\chi}_i-\gamma \boldsymbol{v}^*_i)dt + \sqrt{2\gamma/\beta}d\boldsymbol{W}_i\\
&\boldsymbol{Q}^{*} =\boldsymbol{J A}^{-1}(\Delta \psi \boldsymbol{d}+\boldsymbol{b}+\bm{\xi})
\end{aligned}
\right.
\label{langevin_rbe}
\end{equation}
with the same initial values as $(\{\bm{r}_i\},\{\bm{v}_i\},\bm{Q})$. Here $\bm{\xi}=(\xi_1,\cdots,\xi_{N_{\text{ele}}})^{T}$ denotes the fluctuation vector of electrode charge distribution. By the Euler-Maruyama scheme as the numerical integrator, the discrete update for the configuration reads
\begin{equation}
\left\{
\begin{aligned}
(\boldsymbol{r}_i^*)^{n+1}&=(\boldsymbol{r}_i^*)^{n}+(\boldsymbol{v}_i^*)^{n}\Delta t\\
(\boldsymbol{v}_i^*)^{n+1}&=(\boldsymbol{v}_i^*)^{n}+m_i^{-1}\left[\boldsymbol{F}_i^n+\bm{\chi}_i-\gamma (\boldsymbol{v}_i^*)^{n}\right]\Delta t + m_i^{-1} \sqrt{2\gamma/ \beta} \Delta \boldsymbol{W}_i^n\\
(\boldsymbol{Q}^*)^{n+1}&=\boldsymbol{J A}^{-1}(\Delta \psi \boldsymbol{d}+\boldsymbol{b}^n+\bm{\xi}),\\
\end{aligned}
\right.
\label{dis_SDE}
\end{equation}
where $\Delta \boldsymbol{W}_i^n=\boldsymbol{W}_i^{n+1}-\boldsymbol{W}_i^{n}$. To measure the difference between two configuration distributions provided by Eqs.~\eqref{langevin_ori} and \eqref{dis_SDE}, the Wasserstein-2 distance \cite{santambrogio2015optimal} is used
\begin{equation}
\mathscr{W}_2(\mu, \nu)=\left(\inf _{\gamma \in \Pi(\mu, \nu)} \int_{\mathbb{R}^d \times \mathbb{R}^d}|x-y|^2 d \gamma\right)^{1 / 2},
\end{equation}
where $\Pi(\mu, \nu)$ means all the joint distributions whose marginal distributions are $\mu$ and $\nu$, respectively. By evolving the same initial configuration $(\{\boldsymbol{r}_i^{0}\}, \{\boldsymbol{v}_i^{0}\}, \boldsymbol{Q}^{0})$ 
for a finite time $T$ under both the discrete dynamics in Eq.~\eqref{dis_SDE} 
and the original Langevin dynamics in Eq.~\eqref{langevin_ori}, 
the Wasserstein-2 distance between their terminal configuration distributions is characterized by Theorem~\ref{thm::evolution}. 
It is worth emphasizing that, in contrast to all-atom simulations, 
the electrode charge $\boldsymbol{Q}$ in CP simulations also exhibits randomness under the random batch method. 
Nevertheless, the proof follows a similar argument as in \cite{liang2023random,jin2020random,jin2021convergence,jin2022random}, 
demonstrating a half-order convergence of the configurations with respect to time evolution.

\begin{mythe}
\label{thm::evolution}
 Let $\boldsymbol{Y}(\bm{x}, \cdot)$ be the transition probability of the Eq.~\eqref{langevin_ori} and $\boldsymbol{Y}^*(\bm{x}, \cdot)$ be the transition probability computed by the random batch discretization in the Eq.~\eqref{dis_SDE}, where $\bm{x}$ denotes the initial configuration of the system. Assume the masses $\{m_i\}$ are bounded. If the forces $\{\boldsymbol{F}_i\}$ and vector $\bm{b}$ are bounded, Lipschitz continuous with $\mathbb{E} \bm{\chi}_i=\bm{0}$ and $\mathbb{E} \bm{\xi}=\bm{0}$, then for any $T>0$, there exists a constant $C=C(N_{ion},N_{ele},T)>0$ and $D\ge 1$ such that
\begin{equation}
\sup_{\bm{x}} \mathscr{W}_2(\boldsymbol{Y}(\bm{x}, \cdot), \boldsymbol{Y}^*(\bm{x}, \cdot)) \leq C \sqrt{\Lambda \Delta t+D\Delta t^2},
\end{equation}
where $\Lambda$ is an upper bound for $\left|\left|\mathbb{E}\left|\boldsymbol{\chi}_i\right|^2\right|\right|_{\infty}+\left|\left|\mathbb{E}\left|\boldsymbol{\xi}_j\right|^2\right|\right|_{\infty}$.
\end{mythe}

Finally, the total computational complexity of the SRBE method per simulation step is composed of two main contributions: the calculation of electrolyte ion dynamics and the adjustment of electrode charges. For the electrolyte ion dynamics, the complexity of the near part is 
\begin{equation}
C_{\mathcal{R}}=\frac{4}{3}\pi r_c^3 \rho_{\text{ion}}N_{\text{ion}}, 
\end{equation}
where $r_c$ is the cutoff radius. Since the ion charge structure factor $S_1(\bm{k})$ are computed once for $2M+P$ modes and then reused for all $N_{\text{ion}}$ ions, the complexity of this Fourier part is
\begin{equation}
C_{\mathcal{F}}=2MN_{\text{ion}}+PN_{\text{ion}}.
\end{equation}
Analogously,  $S_1(\boldsymbol{k})$ is computed in the electrolyte ion dynamics and can be used. The complexity of the electrode charge is 
\begin{equation}
C_{\mathcal{C}}=2MN_{\text{ele}}+PN_{\text{ele}}.
\end{equation}
Typically, one chooses $\alpha=\rho_{\text{ion}}^{1/3}$, then the complexity of near part $C_{\mathcal{N}}$ is $O(N_{\text{ion}})$ \cite{jin2021random}. The truncation term is selected by $M=\lambda\mu$ less than the sampling number $P$. Meanwhile, by the variance analysis of Eqs.~\eqref{eq::charge_Var} and \eqref{eq::F_Var_estimate}, $P$ can be selected to be an $O(1)$ constant independent of the total number of particles $N=N_{\text{ion}}+N_{\text{ele}}$ once the density is fixed \cite{jin2021random,liang2022superscalability,irbe2022jpca}. The overall computation complexity of the SRBE method becomes 
\begin{equation}
C_{\text{all}}=C_{\mathcal{R}}+C_{\mathcal{F}}+C_{\mathcal{C}}=O(N).
\end{equation}

Moreover, the proposed SRBE method offers two primary advantages for simulating electrode/ionic liquid systems. First, it substantially reduces the variance of electrode charge fluctuations in highly anisotropic systems, where the scale of a certain dimension is much larger than the others. Second, it inherits the advantage of the RBE with a total computational complexity of $O(N)$, ensuring high efficiency and scalability in large-scale simulations.

\section{Results and discussions}
\label{sec::result}
In this section, we conduct numerical experiments under constant potential conditions and perform a comparative analysis against the PPPM algorithm \cite{ahrens2021constant}. The evaluation focuses on quantifying both numerical accuracy and computational efficiency across these methods. All simulations were carried out using our method implemented in LAMMPS \cite{thompson2022lammps} (version 23June2022), and executed on the ``Siyuan Mark-I"
cluster at Shanghai Jiao Tong University, which consists of 936 nodes, each equipped with 2 $\times$ Intel Xeon ICX Platinum 8358 CPUs (2.6 GHz, 32 cores) and 512 GB of memory.

The simulated system consists of a supercapacitor with graphene as the electrode material and a coarse-grained ionic liquid ($[\text{BMim}^+][\text{PF}_6^{-}]$) as the electrolyte solution \cite{merlet2013simulating}. The length, width and height of the unit cell are 32.3 $\textup{\AA}$, 34.4 $\textup{\AA}$, and 136 $\textup{\AA}$, respectively, and it contains 320 $[\text{BMim}^+][\text{PF}_6^{-}]$ molecules and 3136 carbon atoms. The short-range Ewald part and  the Lennard-Jones (LJ) interaction is calculated by the Verlet list method \cite{verlet1967computer}, with a cutoff distance of $r_c=16~\textup{\AA}$. The long-range interaction is the electrostatic interaction, where the PPPM algorithm and the SRBE method are used for comparison and verification respectively. The bandwidth $\eta$ is 0.556 $\textup{\AA}$ and splitting parameter $\alpha$ is 0.187 $\textup{\AA}^{-1}$. To ensure the accurate calculation of charges of electrode atoms in the CP simulation, the calculation limit error selected for the PPPM algorithm is  $10^{-6}$ \cite{sitlapersad2024simple,langford2022constant,lin2024molecular}. The vacuum layer zero-padding ratio in the non-periodic direction ($z$-direction) is set to be $\lambda=3$, thus the aspect ratio of the selected electrode/ionic liquid system is roughly $[\lambda\mu]=13$, where $[\cdot]$ denotes the floor function. The Shake constraint algorithm \cite{krautler2001fast} is applied to maintain the rigid structure of $[\text{BMim}^+][\text{PF}_6^{-}]$ molecules during the CP simulation. The simulated system is under the canonical ensemble, maintaining the temperature of the electrolyte at $500~K$, and the unit time step is $\Delta t=1~fs$. A pre-equilibrium of $2\times 10^6$ steps ($2~ns$) is firstly carried out, followed by a statistical equilibrium of $3\times 10^6$ steps ($3~ns$).

\subsection{Accuracy performance}
\label{sec::accuracy}

For the supercapacitor system under consideration, we first investigate the effect of the truncation parameter 
$M$ for low-frequency Fourier modes on the accuracy. A potential difference of $\Delta\psi=2~V$ is applied across the two electrodes. We consider the electrode charge at the equilibrium stage as well as the evolution curve of the atomic charge on the electrode. For the SRBE method, we examine the effect of the truncation parameter $M = 0, 2, 5, 13,$ and $20$, and use batch sizes $P = 50, 100,$ and $200$ for sampling high-frequency Fourier modes. The results are summarized in Table~\ref{cpm_srbe_charge}, and the reference value and standard deviation of the equilibrium charge on the negative electrode are obtained using the PPPM method, yielding $-3.233$ and $0.266$, respectively. The data in the table demonstrate that, regardless of the chosen batch size for sampling high-frequency Fourier modes, increasing the truncation parameter $M$ from $0$ to $13$ leads to charge means and fluctuations that increasingly align with the reference values. When $M$ exceeds the aspect ratio $[\lambda\mu] = 13$, further increasing $M$ yields no significant improvement in accuracy, in consistent with the theoretical analysis in Section~\ref{sec::SRBE}. Notably, the case $M = 0$ degenerates to the classical RBE method without the symmetry-preserving technique. This setting performs poorly in capturing the charge distribution, especially in small-batch scenarios, where large fluctuations may prevent stable simulations. These results highlight the critical importance of incorporating symmetry-preserving technique in the development of stochastic algorithms for electrode/ionic liquid systems. Figure~\ref{srbe_pppm_charge_cpm} shows the dynamic atomic charge results, where the SRBE method is also capable of accurately capturing the system’s dynamic behavior. In contrast, using a smaller truncation parameter $M$, leads to significant fluctuations during time evolution, which hinders the accurate representation of the system properties.

\begin{table}[htbp]  \caption{The average value (Ave) and standard deviation (SD) of negative electrode charge at equilibrium stage by the SRBE method of different truncation terms. ``/" indicates that the corresponding configuration fails to achieve a stable CP simulation.}

\centering
	\vspace{1.0em}
	\begin{tabular}{lcccccc}
		
		\toprule[1.5pt]
		\vspace{0.3em}
		& $M$ &$0$& $2$ & $5$ & $13$ & $20$	 \\
	   \midrule[1pt]
       $P=50$& Ave	& /&-4.791 & -3.847 & -3.377 & -3.340     \\
		
	&SD	& /&0.848 &  0.387 & 0.241  &0.260     \\
		
	$P=100$& Ave	& /&-4.462 & -3.665 & -3.311 & -3.357     \\
		
	&SD	& /&0.774 &  0.377 & 0.240  &0.275     \\
 $P=200$&Ave	& -32.124&-3.834 & -3.427 & -3.358  &-3.412    \\
		
	&SD	& 27.626& 0.566 &  0.288 &0.251  &0.263     \\
		
		\bottomrule[1.5pt]
		
	\end{tabular}
	
	\label{cpm_srbe_charge}
	
\end{table}

\begin{figure}[h!]
		\centering
\includegraphics[width=0.6\linewidth]{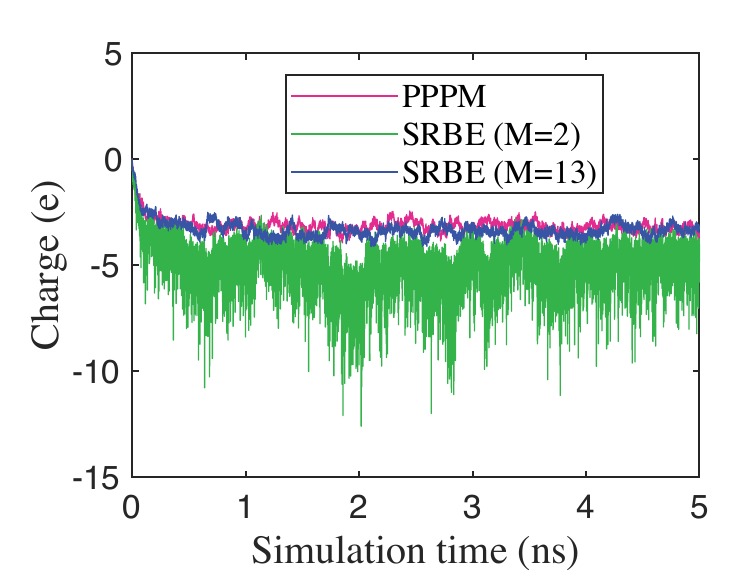}
		\caption{The evolution profile of atomic charge on the negative electrode by the SRBE method ($P=50$) with different truncation parameters $M=2$ and $M=13$, compared to the PPPM method.}
		\label{srbe_pppm_charge_cpm}
\end{figure}

Next, we examine accuracy as a function of the high-frequency batch size $P$ in this system. To ensure accuracy, the truncation parameter is fixed as $M=13$. For different batch sizes $P = 20, 50, 100$, and $200$, the mean and variance of the per-atom energy are shown in Figure~\ref{srbe_P_pppm_charge_cpm}. When low-frequency Fourier modes are accurately computed by the SRBE, a batch size of $P \ge 50$ is sufficient to reliably reproduce both the system energy and its fluctuation. To further assess the SRBE method in capturing more challenging physical quantities specific to the electrode/ionic liquid system, such as the distribution of free ion density in supercapacitors and the electric double layer structures near both electrodes, we conducted simulations using SRBE with $P=50$. The results are presented in Figure~\ref{srbe_pppm_density_cpm}. One observes that, regardless of whether a low voltage ($0.2~V$, (a)(d)), moderate voltage ($2~V$, (b)(e)), or high voltage ($5~V$, (c)(f)) is applied to the two electrodes, the SRBE method with this parameter setting achieves accuracy comparable to that of the traditional PPPM approach. This demonstrates both the applicability and reliability of the newly designed method for such simulations.

\begin{figure}[h!]
		\centering
\includegraphics[width=0.6\linewidth]{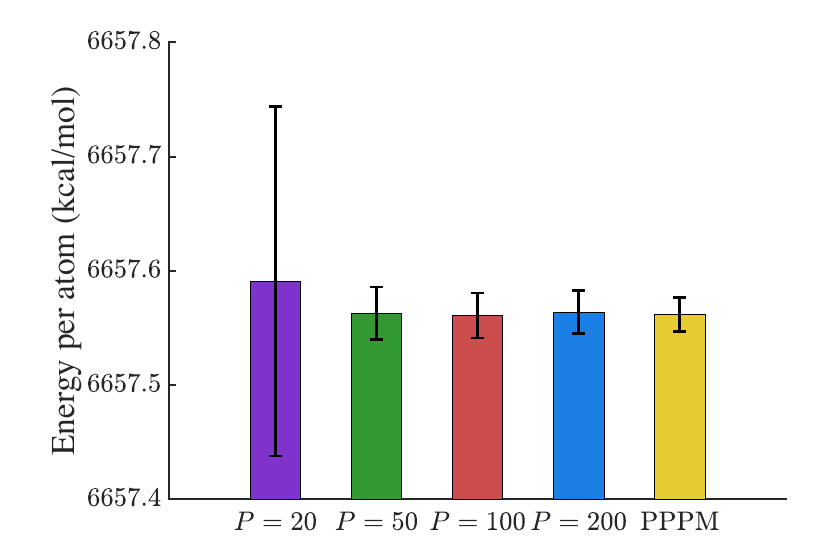}
		\caption{The total energy per atom and the errorbar at equilibrium stage by the SRBE algorithm with different batch sizes $P=20,50,100$ and $200$, compared to the PPPM method.}
\label{srbe_P_pppm_charge_cpm}
\end{figure}

\begin{figure}[h!]
		\centering
\includegraphics[width=1.0\linewidth]{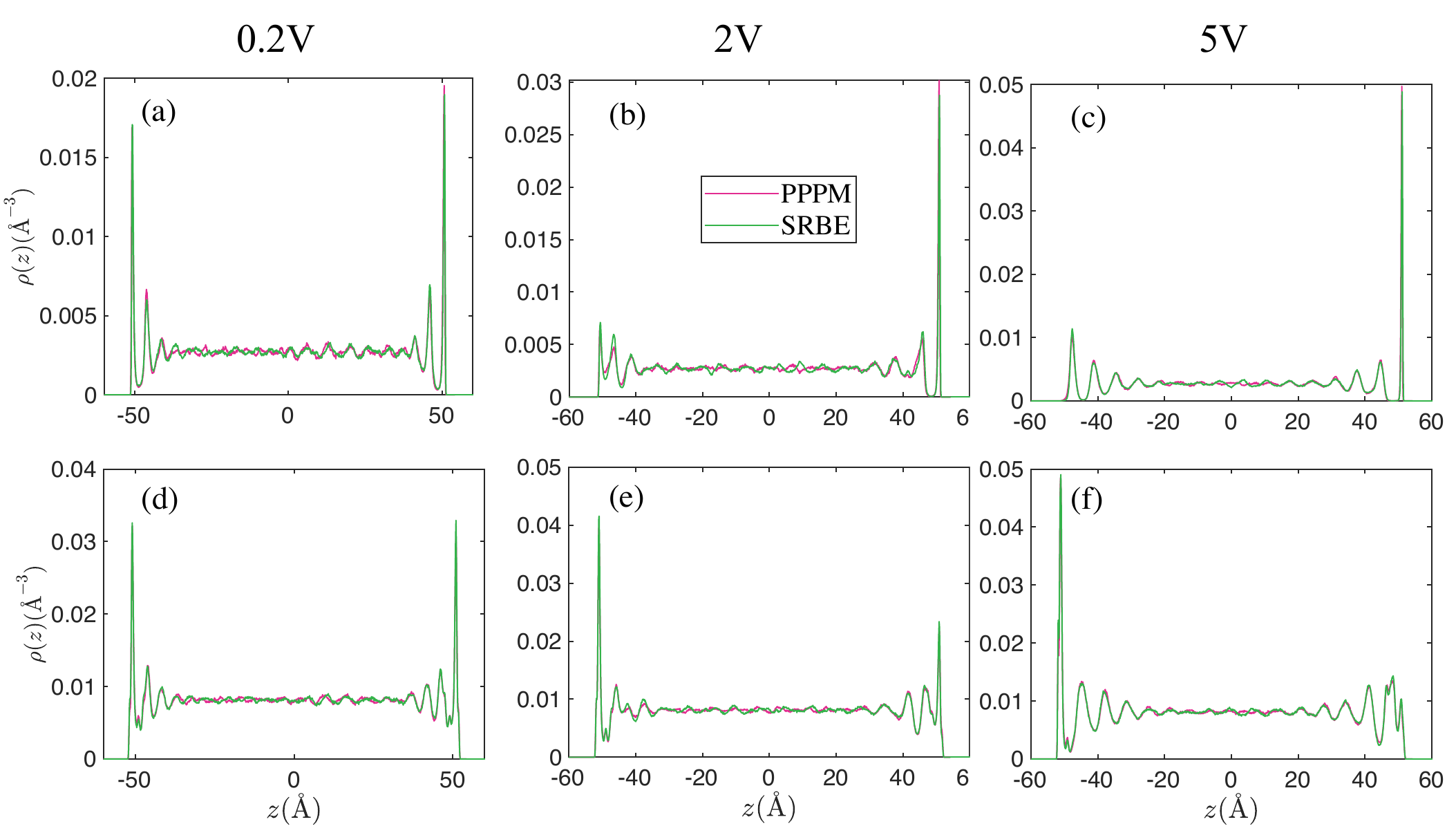}

		\caption{The number density distribution profile $\rho$ of the cation (a-c), anion  (d-f) in solution along the direction of $z$ at a voltage of $0.2~V$, $2~V$ and $5~V$. The long-range part of charge adaption vector $\bm{b}$ is calculated by the SRBE method with $P=50$ and $M=13$, compared to the PPPM method.}
\label{srbe_pppm_density_cpm}
\end{figure}

\subsection{Time performance}
\label{sec::time}

In this subsection, we test the efficiency for the long-range part of the SRBE algorithm. To access a fair comparison, parameters of the SRBE method is set to be $P=50$ and $M=13$, which achieves the same accuracy level as the PPPM method, as presented in Section~\ref{sec::accuracy}. We use the elapsed wall-clock time to evaluate the performance of different algorithms, and estimate the average CPU time per step by running 1000 simulation steps.

First, we investigate how the time scaling of the SRBE method changes with increasing system size $N$. Since anisotropy plays a crucial role in CP simulations of electrode/ionic liquid systems, the system size is scaled while keeping the aspect ratio $\lambda\mu$ fixed. Figure~\ref{srbe_P_pppm_srbe_P50_orderN} presents a comparison of the time consumption between the PPPM and SRBE methods executed on 64 CPU cores, where the number of particles $N$ is scaled from $3\times 10^4$ to nearly $2\times 10^6$. The solid lines represent linear fittings of the data on a log-log scale, while the dashed lines indicate the trend of $\mathcal{O}(N)$ growth for slope comparison. The test results clearly demonstrate that the SRBE method exhibits an efficient linear complexity with a considerable computation savings for the long-range contribution, highlighting the appealing performance of the algorithm.

\begin{figure}[h!]
		\centering
\includegraphics[width=0.6\linewidth]{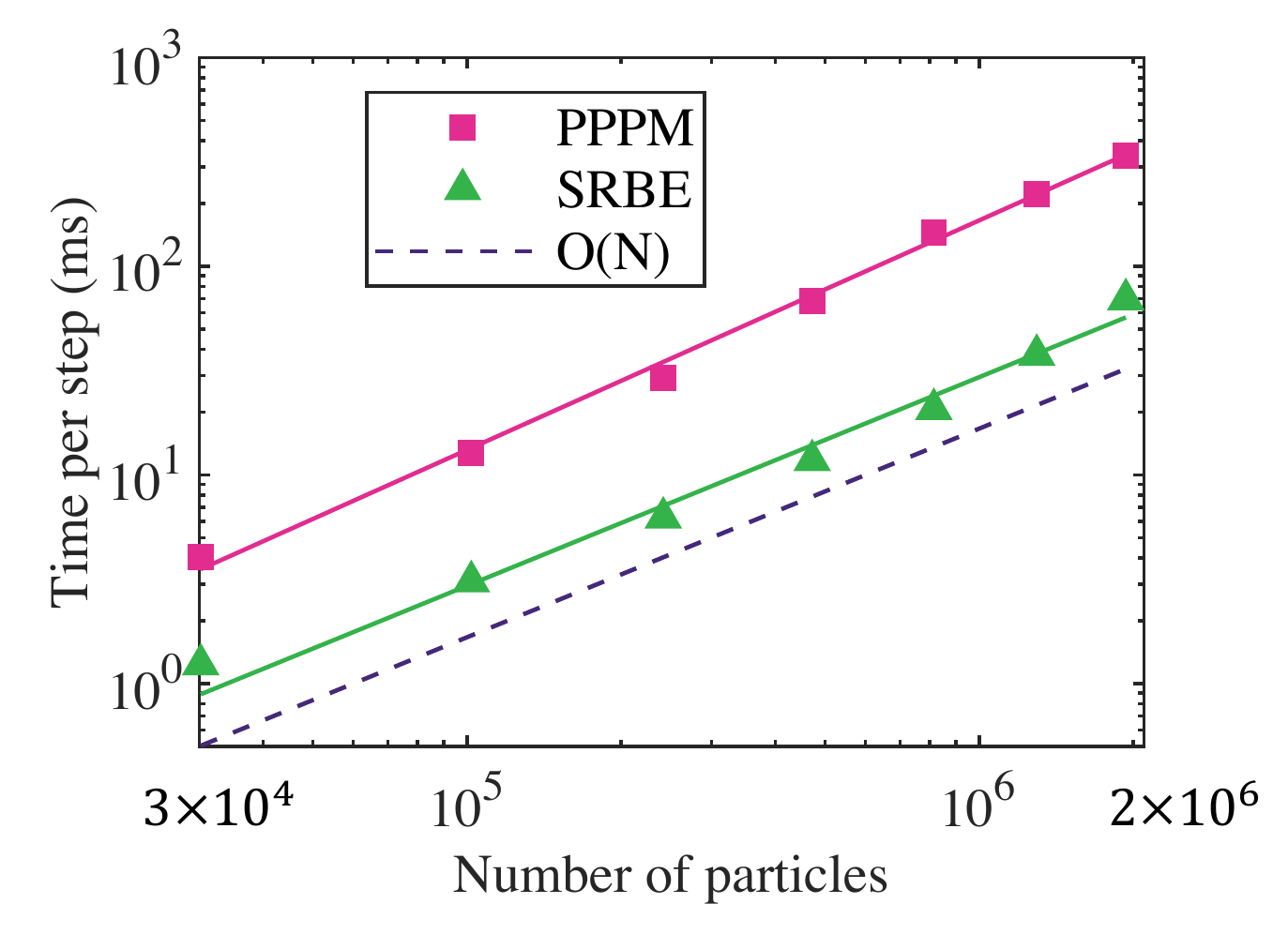}
		\caption{The average calculation time per step  of the long-range interaction part
with increasing total number of particles $N$. The number of CPU cores is fixed with $\zeta=64$. }
\label{srbe_P_pppm_srbe_P50_orderN}
\end{figure}

Another key criterion for evaluating whether an algorithm is suitable for large-scale parallel MD simulations is the scalability, which quantitatively reflects the communication overhead of the method. In the time experiments of this subsection, we evaluate both the weak and strong scalability of the SRBE method and compare its performance with the mainstream FFT-based PPPM method. Weak scalability is evaluated by keeping the average system size per processor constant while proportionally increasing both the total system size and the number of processors, and observing the resulting change in computation time. The weak scalability  is defined by
$\eta^{\text{weak}}(\zeta)=T_{\text{min}}/T(\zeta),$
where $\zeta$ denotes the number of processors, $T(\zeta)$ represents the corresponding time consumption, and $T_{\text{min}}=T(\zeta_{\text{min}})$ denotes the time consumption of the minimal number of processors $\zeta_{\text{min}}$. The value of this metric closer to $1$ indicates better weak scalability for the method \cite{gustafson1988reevaluating}.
 By fixing the system size per core at $N_{\text{core}} = 3776$ and utilizing up to $\zeta_{\text{max}} = 343$ CPUs, the time consumption and weak scalability of the SRBE method compared to the PPPM method are presented in Figure~\ref{srbe_pppm_time_strong_weak}(a-b). The SRBE method consistently maintains over $60\%$ weak scalability for the long-range part, while the PPPM method drops to nearly $10\%$ under parallel computation. 
 
Different from the weak scalability, strong scalability $\eta^{\text{strong}}(\zeta)$ is defined under the condition of a fixed total system size, by increasing the number of processors $\zeta$ to evaluate the corresponding performance,
$\eta^{\text{strong}}(\zeta)=(T_{\text{min}}\zeta_\text{min})/(\zeta T(\zeta)), $
where similarly, the value closer to 1 indicates a stronger advantage of the method in large-scale parallel computing. Figure~\ref{srbe_pppm_time_strong_weak}(c-d) present a comparison of the strong scalability between the SRBE and PPPM methods, where the system size is fixed at $N=101952$ and $\zeta$ from $1$ to $512$. It is observed that in the computation of the long-range part increases, the SRBE method consistently outperforms the PPPM method in terms of computational time and strong scalability. As $\zeta$ increases, this advantage of the SRBE method becomes increasingly pronounced, reaching over two orders of magnitude improvement in scalability at $\zeta = 512$ compared to the PPPM method, which is dropped below $1\%$ due to the use of communication-intensive FFT. These two scalability results highlight the superior parallel performance of the symmetry-preserving random batch strategy for large-scale CP simulations in electrode/ionic liquid systems, demonstrating its strong potential for applications in fields such as supercapacitor design.

\begin{figure}[h!]
		\centering
\includegraphics[width=1.0\linewidth]
{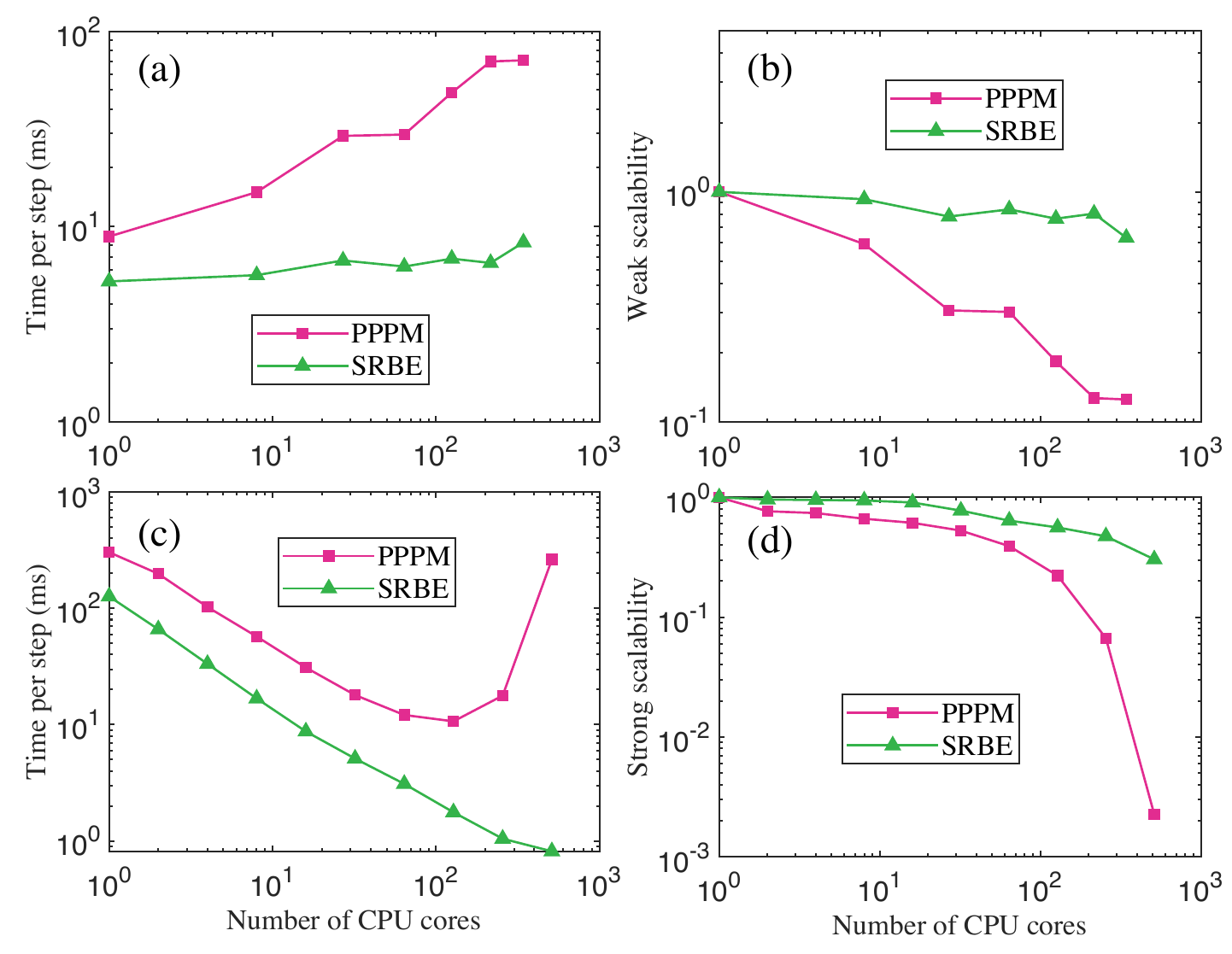}
		\caption{The average calculation time per step (a) and weak scalability (b) of the long-range interaction part with the fixed average number of particles per CPU core $N_{\text{core}}=3776$. The average calculation time per step (c) and strong scalability (d) of the long-range interaction part with fixed total number of particles $N=101952$. The parameters of the SRBE method are selected by $P=50$ and $M=13$.}
	\label{srbe_pppm_time_strong_weak}
\end{figure}

\section{Conclusion}
\label{sec::conclusion}
We have developed a novel SRBE method to accelerate CP simulations, which are critical for advancing the understanding of electrochemical interface phenomena.
 This algorithm builds upon the RBE method and incorporates the SPMF condition to accelerate the update of long-range part of the electrode charge fluctuations and forces acting on particles. Through rigorous theoretical analysis, it is demonstrated that selecting appropriate truncation terms can effectively eliminate charge fluctuation errors caused by random sampling, while ensuring the overall computational complexity remains $O(N)$. Numerical results confirm that the SRBE algorithm can accurately reproduce the discharge process, the total energy fluctuations of the system at equilibrium, as well as the electric double-layer structure near the electrode. In terms of efficiency, the SRBE algorithm exhibits one to two orders of magnitude improvement in multicore scalability compared with the PPPM method.
 These results strongly demonstrate the accuracy and efficiency of the SRBE algorithm, highlighting its great potential for large-scale constant-potential simulations in practical applications such as nanoporous electrode design.

\section*{Acknowledgments}
This work is funded by the National Natural Science Foundation of China (grants No. 12426304, 12325113 and 12350710181), and the Science and Technology Commission of Shanghai Municipality (grant No. 23JC1402300). The authors also acknowledge the support from the SJTU Kunpeng \& Ascend Center of Excellence.

\section*{Declarations}
\subsection*{Conflict of interest}
The authors declare that they have no conflict of interest.

\subsection*{Data availability}
The data that support the findings of this study are available from the corresponding author upon reasonable request.

\appendix
\section{The Debye-H\"{u}ckel theory}
\label{app::DH}

This appendix provides the error bounds of
\begin{equation}
    \label{eq::error_structure}
    \mathcal{G}_1(\bm{k},\bm{r}_i)=\sum_{j=1}^{N_{\text{ion}}}q_je^{i\bm{k}\cdot\bm{r}_{ij}},\quad \mathcal{G}_{2}(\bm{k},\bm{r}_i)=\sum_{j=1}^{N_{\text{ele}}}Q_je^{i\bm{k}\cdot(\bm{r}_i-\bm{R}_j)}
\end{equation}
under the DH theory \cite{hansen2013theory}.
First, we consider the contribution of the electrolyte structure factor, denoted as $\mathcal{G}_1(\bm{k},\bm{r}_i)$. For convenience, we examine a system of $N$ hard-sphere ions with charges $\pm q$ and radius $a>0$. Fixing one ion with charge $q$ at the origin, we then analyze the distribution of the surrounding ions. Inside the excluded hard-sphere region $\mathring{B}(a)=\{0<r<a\}$, no other ions are present; hence, the potential function satisfies the Laplace equation that $-\nabla^2\phi=0$. Outside the excluded region $\mathring{B}(a)=\{0<r<a\}$, the charges are approximately distributed according to the Boltzmann distribution $\rho_{\pm}(\bm{r})=\pm q\rho_{\pm,\infty}e^{\mp \beta q\phi}$, where $\rho_{\pm,\infty}=N/2V=\rho/2$ and $\beta=1/k_BT$ denotes the Boltzmann constant. Hence the linearized Poisson equation reads
\begin{equation}
\label{eq::Debye}
    -\nabla^2\phi=q\rho_{\infty,+}e^{-\beta q\phi}-q\rho_{\infty,-}e^{\beta q\phi}\approx -\kappa^2\phi\triangleq\rho(\bm{r}),
\end{equation}
where $\kappa=\sqrt{\beta q^2\rho}$ is the inverse of Debye length $\lambda_D$. The solution to Eq.~\eqref{eq::Debye} gives  
\begin{equation}
    \phi(r)= \begin{cases}\dfrac{q}{4 \pi  r}-\dfrac{q \kappa}{4 \pi (1+\kappa a)}, & r<a, \\ \dfrac{q e^{\kappa a} e^{-\kappa r}}{4 \pi r(1+\kappa a)}, & r\geq a ,\end{cases}
\end{equation}
which provides the bound of $\mathcal{G}_1(\bm{k},\bm{r}_i)$ as
\begin{equation}
    \label{eq::bound_G1}
    \begin{aligned}
        |\mathcal{G}_1(\bm{k},\bm{r}_i)|&\approx\left|q+\int_{\mathbb{R}^3\backslash\mathring{B}(a)}\rho(\bm{r})e^{i\bm{k}\cdot\bm{r}}\right|\\
        &\le q\left|1-\frac{1}{(1+\kappa a)(1+k^2/\kappa^2)}\left(\frac{\kappa}{k}\sin(ka)+\cos(ka)\right)\right|.
    \end{aligned}
\end{equation}
Eq.~\eqref{eq::bound_G1} shows a uniform upper bound of $|\mathcal{G}_1(\bm{k},\bm{r}_i)|\le C_{\text{DH}}q$, where the constant $C_{\text{DH}}$ is independent of ion number $N_\text{ion}$ or Fourier mode $\bm{k}\in \mathcal{L}_k\backslash\mathcal{I}$. 

For $\mathcal{G}_{2}(\bm{k},\bm{r}_i)$, since the charges on the electrode are fixed at uniformly distributed positions, their distribution can be regarded as a lattice discretization of the continuous Boltzmann distribution. Consequently, the same analytical considerations as in Eq.~\eqref{eq::bound_G1} apply, and a similar uniform bound holds, namely, $|\mathcal{G}_{2}(\bm{k},\bm{r}_i)| \leq \widetilde{C}_{\text{DH}} q$, where $\widetilde{C}_{\text{DH}} $ is independent of electrode atom number $N_\text{ele}$  or Fourier mode $\bm{k}\in \mathcal{L}_k\backslash\mathcal{I}$.


\end{document}